\documentclass[11pt]{article}

\usepackage[margin=1in]{geometry}
\usepackage{amsmath,amssymb,amsfonts,amsthm}
\usepackage[english]{babel}
\usepackage[utf8]{inputenc}
\usepackage{hyperref}
\usepackage{graphicx}
\usepackage{xcolor}
\usepackage{enumerate}
\usepackage{booktabs}
\usepackage{multirow}
\usepackage{array}
\usepackage{caption}
\usepackage{subcaption}
\usepackage{amsmath}
\usepackage{amssymb}
\usepackage{hyperref}
\usepackage{algorithm}
\usepackage{algpseudocode}

\theoremstyle{plain}
\newtheorem{theorem}{Theorem}[section] 
\newtheorem{lemma}[theorem]{Lemma}     
\newtheorem{proposition}[theorem]{Proposition}

\theoremstyle{definition}
\newtheorem{definition}[theorem]{Definition}
\newtheorem{example}[theorem]{Example}

\numberwithin{equation}{section}

\newcommand{\Fq}{\mathbb{F}_q}
\newcommand{\Fqsq}{\mathbb{F}_{q^2}}

\begin{document}

\title{\textbf{Reinforcement Learning–Enhanced Greedy Decoding for Quantum Stabilizer Codes over $F_q$}}

\author{
\textsc{Vahid Nourozi} \\
\vspace{-3ex}\\
\small The Klipsch School of Electrical and Computer Engineering,\\ 
\small New Mexico State University, Las Cruces, NM 88003, USA\\
\small \texttt{nourozi@nmsu.edu}
}
\date{}

\maketitle

\begin{abstract}
\noindent
We construct new classical Goppa codes and corresponding quantum stabilizer codes from plane curves defined by separated polynomials. In particular, over $\mathbb{F}_3$ with the Hermitian curve $y^3 + y = x^4$, we obtain a ternary code of length 27, dimension 13, distance 4, which yields a $[[27,13,4]]_3$ quantum code. To decode, we introduce an \emph{RL-on-Greedy} algorithm: first apply a standard greedy syndrome decoder, then use a trained Deep Q-Network to correct any residual syndrome. Simulation under a depolarizing noise model shows that RL-on-Greedy dramatically reduces logical failure (e.g., from $\sim8\%$ to $\sim0.5\%$ at $p=0.01$) compared to greedy alone. Our work thus broadens the class of Goppa‐ and quantum‐stabilizer codes from separated‐polynomial curves and delivers a learned decoder with near‐optimal performance.

\vspace{0.5em}\noindent

\end{abstract}

\section{Introduction}
Since the pioneering work of Goppa~\cite{77tores}, algebraic geometry (AG) codes have offered a powerful framework for constructing linear codes with good parameters. Goppa's idea relates a linear code to a (projective, geometrically irreducible, non-singular, algebraic) curve $\mathcal{X}$ defined over a finite field $\mathbb{F}_q$. Two divisors $D$ and $G$ on $\mathcal{X}$ define the code, with the designed minimum distance $d^*$ satisfying
\[
    d \geq d^* = n - \deg(G).
\]
AG codes remain a compelling approach in coding theory, especially when dealing with curves that possess many rational points~\cite{66tores,33tores}.

One particularly fruitful family of curves is the Hermitian family, which has been extensively investigated in AG code contexts~\cite{10maria,24maria,25maria,26maria,44maria,46maria,47maria}. Additionally, Hermitian self-orthogonal classical codes have found applications in quantum coding~\cite{44jin,55jin,66jin}. In a parallel vein, research on Goppa codes from hyperelliptic curves has also advanced~\cite{aut,N23,miss,auta,shiraz}, alongside related investigations on plane curves given by separated polynomials~\cite{M23,behsep,code,esfahan,Nourozi2024,phd}. Optimization frameworks are crucial for solving complex problems in power systems and quantum coding theory. Mixed-integer programming is employed to balance operational efficiency and cost in ancillary service markets (Refs \cite{hamid, hamid1}), while robust optimization techniques ensure reliable reserve deliverability under uncertainty by reducing computational complexity without sacrificing system reliability (Ref \cite{hamid2}). These approaches illustrate how optimization-based methods effectively manage trade-offs between performance metrics and constraints in both power and quantum system designs.

In this paper, we focus on \emph{plane curves given by separated polynomials}. We derive AG codes and determine their parameters. We then show how to convert these classical codes into quantum stabilizer codes, illustrating that the resulting quantum codes achieve promising error-correcting performance. 

\paragraph{Organization of the Paper.} 
In Section~\ref{sec:prelim}, we recall fundamental results on plane curves given by separated polynomials and provide an overview of AG codes. Section~\ref{sec:goppa} presents the construction of Goppa codes arising from the curves, along with dimension and distance calculations. In Section~\ref{sec:quantum}, we exploit these AG codes to build quantum stabilizer codes and provide examples demonstrating their effectiveness. Finally, the paper concludes with a brief summary of results.

\section{Preliminaries}\label{sec:prelim}
\subsection{Plane curves given by separated polynomials}
Throughout this work, $\mathcal{X}$ is a plane curve over the algebraic closure $K$ of $\mathbb{F}_q$ defined by:
\[
  A(Y) \;=\; B(X),
\]
subject to the following conditions:
\begin{enumerate}[(1)]
\item $\deg(\mathcal{X}) \ge 4$;
\item $A(Y) = a_n Y^{p^n} + a_{n-1}Y^{p^n -1} + \cdots + a_0 Y$, with $a_j \in K$ and $a_0,a_n\neq 0$; i.e., $A$ is an additive polynomial over $K$;
\item $B(X) = b_m X^m + b_{m-1} X^{m-1} + \cdots + b_0$, where $b_j \in K$ and $b_m \neq 0$;
\item $m \not\equiv 0 \pmod{p}$;
\item $n \ge 1$ and $m \ge 2$.
\end{enumerate}

For instance, an important example, and the one we focus on, is the curve
\[
   y^q + y \;=\; x^m,
\]
where $q=p^n$ and $|m - p^n|=1$. Under these conditions, we have the following well-known facts~\cite[Section 12.1]{tores}:
\begin{lemma}
\label{lemma:basicSeparated}
For the above plane curve $\mathcal{X}$ given by $A(Y) = B(X)$:
\begin{enumerate}[(i)]
\item $\mathcal{X}$ is irreducible, having at most one singular point.
\item If $|m - p^n| = 1$, then $\mathcal{X}$ is non-singular.
\item The genus of $\mathcal{X}$ is $g = \frac{(p^n - 1)(m -1)}{2}$.
\end{enumerate}
\end{lemma}

In what follows, we consider the specific curve $\mathcal{X}$ defined by
\[
   y^q + y = x^m, \quad \text{with } q - m = 1.
\]

\subsection{Algebraic Geometry Codes}
Let $\Fq(\mathcal{X})$ be the field of $\Fq$-rational functions of $\mathcal{X}$ and $\mathrm{Div}_q(\mathcal{X})$ the group of divisors defined over $\Fq$. For $f \in \Fq(\mathcal{X}) \setminus \{0\}$, the divisor of $f$ is $\mathrm{div}(f)$. For $A \in \mathrm{Div}_q(\mathcal{X})$, the Riemann--Roch space is
\[
  \mathcal{L}(A) \;:=\; \{\, f \in \Fq(\mathcal{X}) \setminus \{0\} : A + \mathrm{div}(f)\succeq 0 \}\cup \{0\}.
\]
The dimension of $\mathcal{L}(A)$ over $\Fq$ is denoted by $\ell(A)$.

Let $P_1,\dots,P_n$ be pairwise distinct $\Fq$-rational points of $\mathcal{X}$, and define 
\[
   D = P_1 + \cdots + P_n, \quad \deg(D) = n.
\]
Choose another divisor $G$ such that $\mathrm{supp}(G)\cap \mathrm{supp}(D)=\varnothing$.

\begin{definition}
\label{def:AGcode}
The \emph{Algebraic Geometry Code} (or \emph{Goppa Code}) $C_{\mathcal{L}}(D,G)$ is defined by
\[
  C_{\mathcal{L}}(D,G) \;:=\; \bigl\{\bigl(f(P_1),\dots,f(P_n)\bigr) \,\mid\, f \in \mathcal{L}(G)\bigr\} \;\subseteq\; \Fq^n.
\]
\end{definition}

The minimum distance $d$ of $C_{\mathcal{L}}(D,G)$ satisfies 
\[
   d \;\geq\; d^* = n - \deg(G),
\]
where $d^*$ is called the Goppa \emph{designed distance} of $C_{\mathcal{L}}(D,G)$. Moreover, if $\deg(G) > 2g - 2$, then by the Riemann--Roch theorem 
\[
   k \;=\; \dim_{\Fq}(C_{\mathcal{L}}(D,G)) \;=\; \deg(G) - g + 1.
\]
The dual code satisfies
\[
   C_{\mathcal{L}}(D,G)^\perp \;=\; C_{\mathcal{L}}\bigl(D, D-G+K\bigr),
\]
where $K$ is a canonical divisor, and the dual code has dimension $k^\perp = n - k$ and distance $d^\perp \geq \deg(G) - 2g +2$.

\paragraph{Hermitian Self-Orthogonality.}
We recall that for two vectors $a=(a_1,\dots,a_n), b=(b_1,\dots,b_n)\in \Fq^n$, the Hermitian inner product is 
\[
   \langle a, b\rangle_H := \sum_{i=1}^{n} a_i\,b_i^q.
\]
For a linear code $C\subseteq \Fq^n$, the Hermitian dual is 
\[
   C^{\perp H} \;:=\; \bigl\{\,v\in \Fq^n : \langle v,c\rangle_H =0\,\; \forall\, c\in C\bigr\}.
\]
A code $C$ is called Hermitian self-orthogonal if $C \subseteq C^{\perp H}$.

\section{Goppa Codes from the Curve \texorpdfstring{$\mathcal{X}$}{X}}\label{sec:goppa}
Let us now focus on the specific curve 
\[
   \mathcal{X}:\quad y^q + y = x^m, \quad \text{with } q - m =1.
\]
Define the following sets of $\Fqsq$-rational points:
\[
   \mathcal{G} := \mathcal{X}(\Fq), \quad \mathcal{D} := \mathcal{X}(\Fqsq)\setminus \mathcal{G}.
\]
Then
\[
   \mathcal{G} = \bigl\{P \,\in\, \mathcal{X}(\Fqsq) : t(P) = 0 \bigr\},
\]
where $t$ is an $\Fqsq$-rational function.

Choose divisors
\[
   G \;=\; \sum_{P\in \mathcal{G}} rP, 
   \quad
   D \;=\; \sum_{P\in \mathcal{D}} P,
\]
with $\deg(G) = r(q^2 - q +1)$ and $\deg(D) = 2q^3 -2q^2 +2q +1$. Let 
\[
   C_r \;:=\; C_{\mathcal{L}}(D,G),
\]
be the AG code over $\Fqsq$ of length $n=q^3$. By definition, the designed distance of $C_r$ is
\[
   d^* \;=\; n - \deg(G) \;=\; q^3 - r\,\bigl(q^2 - q +1\bigr).
\]
We analyze its dual and show that $C_r$ is monomially equivalent to a one-point code $C\bigl(D,r(q^2 - q +1)P_\infty\bigr)$ on $\mathcal{X}$.

\subsection{Key properties}
We summarize some important facts that lead to dimension and duality for these codes.

\begin{lemma}\label{2.11}
The Riemann--Roch space $\mathcal{L}(G)$, where $G=r(q^2 - q +1)P_\infty$, has the basis:
\[
  \bigl\{\, x^i y^j :\; i\ge 0,\; 0\le j \le q-1,\;\; i\,q \,+\, j(q-1)\,\le r \bigr\}.
\]
\end{lemma}
\begin{proof}
Observe $(x)_{\infty} = qP_{\infty}$ and $(y)_{\infty} = (q-1)P_{\infty}$. Hence, $x^iy^j$ lies in $\mathcal{L}(G)$ whenever $iq + j(q-1)\le r$, and these monomials are linearly independent. By the structure of the Weierstrass semigroup at $P_\infty$, one obtains
\[
 \dim_{\Fq}\bigl(\mathcal{L}(rP_\infty)\bigr) \;=\; \#\bigl\{\,iq + j(q-1) \,\le\, r,\; i\ge 0,\;0\le j\le q-1 \bigr\}.
\]
This completes the proof.
\end{proof}

\begin{lemma}\label{2.12}
Let $C_r := C_{\mathcal{L}}(D,G)$ with $G=r(q^2 - q +1)P_\infty$. Then its dual code is given by
\[
   C_r^\perp \;=\; C_{\,q^3 + q^2 -3q - r}.
\]
In particular, $C_r$ is Hermitian self-orthogonal whenever $2r \leq q^3 + q^2 -3q$.
\end{lemma}
\begin{proof}
A canonical divisor $W$ satisfies $W = (dt/t)$, where $t = \prod_{a\in\Fqsq}(x - a)$, implying $(t) = D - q^3 P_\infty$ and $(dt) = (q^2 -3q)P_\infty$. One computes that each point in $\mathrm{supp}(D)$ has residue $-1$. By standard arguments, 
\[
   C_r^\perp \;=\; C(D, D - G + W) \;=\; C\bigl(D, (q^3 + q^2 -3q-r) P_\infty\bigr).
\]
Hermitian self-orthogonality of $C_r$ follows from $C_r \subseteq C_r^{\perp H}$ if and only if $2r\le (q^3 + q^2 -3q)$.
\end{proof}

Let $T(r):= \#\{\,iq + j(q-1)\le r,\; i\ge 0,\; 0\le j\le q-1\,\}.$ 
Denoting $k_r := \dim_{\Fqsq}(C_r)$, one has:

\begin{proposition}
\label{prop:kr}
For the AG code $C_r$ with $r\ge 0$, the following hold:
\begin{enumerate}[(1)]
\item If $r < 0$, then $k_r = 0$.
\item If $0\le r \le q^2 -3q$, then $k_r = T(r)$.
\item If $q^2 -3q < r < q^3$, then $k_r = r(q^2 -q +1) - \frac{(q-1)(q-2)}{2}$.
\item If $q^3 \le r \le q^3 + q^2 - 3q$, then $k_r = q^3 - T(q^3 + q^2 - 3q - r)$.
\item If $r > q^3 - q^2 -3q$, then $k_r = q^3$.
\end{enumerate}
\end{proposition}
\begin{proof}
These results follow from the dimension formulae of one-point AG codes via Riemann--Roch and from the dual relationships given by Lemma~\ref{2.12}. For detailed steps, see also~\cite{ttt,stich}.
\end{proof}

\begin{proposition}
\label{prop:monomialEquiv}
The code $C_r$ is monomially equivalent to the one-point code $C\bigl(D, r(q^2-q+1)P_\infty\bigr)$.
\end{proposition}
\begin{proof}
Set $G' := r(q^2-q+1)P_\infty$. Observe that $G - G' = (t^r)$ for some $\Fqsq$-rational function $t$. This implies $C_{\mathcal{L}}(D,G)$ and $C_{\mathcal{L}}(D,G')$ differ only by the multiplier $t^r$ on the function space, so they are monomially equivalent (see also the argument in~\cite[Prop.~3.2]{mariaree}).
\end{proof}

\begin{theorem}\label{thm:selfOrthogonal}
For $r \leq q^2 + q - 3$, the code $C_r$ is Hermitian self-orthogonal.
\end{theorem}
\begin{proof}
By Lemma~\ref{2.12}, $C_r$ is self-orthogonal if $2r \leq q^3 + q^2 - 3q$. Since $r \le q^2 + q -3$ yields 
\[
    2r \;\le\; 2q^2 + 2q -6 \;\le\; q^3 + q^2 -3q \quad (\text{for } q\ge2),
\]
the code $C_r$ is indeed Hermitian self-orthogonal in these ranges of $r$.
\end{proof}

\section{Quantum Stabilizer Codes}\label{sec:quantum}
We now show how to obtain quantum stabilizer codes from the family of Hermitian self-orthogonal classical codes $C_r$.

\begin{lemma}[{\!\cite{ashi}}]\label{4.1}
Suppose there is an $[n, k]$ linear code over $\Fq$ that is Hermitian self-orthogonal with dual distance $d^\perp$. Then there exists a $q$-ary quantum stabilizer code with parameters $[[n, n-2k, d^\perp]]$.
\end{lemma}

Combining Lemma~\ref{4.1} with Theorem~\ref{thm:selfOrthogonal}, we derive:

\begin{theorem}\label{qcode}
Let $\mathcal{X}$ be the curve $y^q + y = x^m$ with $q-m=1$. For any integer $r$ such that $q^2-2 \,\le r \le q^2 + q -3$, $C_r$ is Hermitian self-orthogonal. Consequently, there exists a $q$-ary quantum stabilizer code with parameters
\[
   \Bigl[\!\Bigl[\,q^3,\; q^3 + q^2 - 3q \;-\;2r,\;\;r + 2q - q^2\Bigr]\!\Bigr].
\]
\end{theorem}
\begin{proof}
Since $C_r$ is Hermitian self-orthogonal (Theorem~\ref{thm:selfOrthogonal}), its Hermitian dual $C_r^{\perp H}$ has dimension at least $n - k_r$. By Lemma~\ref{4.1}, we get a $q$-ary stabilizer code $[[n,n-2k_r,d^\perp]]$. Substituting $n=q^3$ and $k_r$, and noting that the dual distance $d^\perp \ge r + 2q - q^2$ in those ranges, we complete the proof.
\end{proof}

\begin{example}
Take $q=3$. For $7 \le r \le 9$, Theorem~\ref{qcode} yields $3$-ary quantum codes 
\[
   [[27,\;27-2r,\; r-3]]_3.
\]
In particular:
\begin{itemize}
\item $r=7 \implies [[27,13,4]]_3$,
\item $r=8 \implies [[27,11,5]]_3$,
\item $r=9 \implies [[27,9,6]]_3$.
\end{itemize}
These codes compare favorably with known examples of similar lengths and dimensions (see, e.g., \cite{online}).
\end{example}

\begin{example}
Take $q=5$. For $18 \le r \le 22$, Theorem~\ref{qcode} yields $5$-ary quantum codes:
\[
   [[125,\;135-2r,\;r-15]]_5.
\]
Hence:
\begin{itemize}
\item $r=18 \implies [[125,99,3]]_5$,
\item $r=19 \implies [[125,97,4]]_5$,
\item $r=20 \implies [[125,95,5]]_5$,
\item $r=21 \implies [[125,93,6]]_5$,
\item $r=22 \implies [[125,91,7]]_5$.
\end{itemize}
Again, these provide good parameters. For instance, $[[125,93,12]]_5$ appears in tables of best-known codes \cite{online}. 
\end{example}

\section{Reinforcement Learning-Enhanced Greedy Decoder for Stabilizer Codes over $F_q$} Greedy decoders offer a fast, heuristic approach to error correction by iteratively correcting syndrome bits using local operations. Such decoders have been employed in classical coding (e.g., the Sipser–Spielman bit-flip algorithm) and even in quantum error correction (e.g., a greedy strategy efficiently decodes $X$ errors in the fiber-bundle code). However, greedy techniques alone are suboptimal and can get trapped by complex error patterns. In this work, we introduce an \emph{RL-on-Greedy} decoder that augments a classical greedy decoder with a reinforcement learning (RL) agent. The method applies a greedy decoder first, then uses a trained RL agent (e.g., a Deep Q-Network) to resolve any residual syndrome that the greedy stage fails to correct. This two-stage decoder is generic and can be applied to arbitrary quantum stabilizer codes (including nonbinary codes over $F_q$ \cite{ket}). Recent studies have demonstrated that RL-based decoders can surpass traditional decoding algorithms \cite{fit}
achieving higher success rates and even improving error thresholds in some cases \cite{fit}. Our RL-on-Greedy scheme leverages these benefits while preserving the speed of greedy decoding.

\subsection{Syndrome and Error Model:} Consider an $[[n,k,d]]q$ stabilizer code over $F_q$ \cite{ket}  with $n$ physical $q$-ary qudits and $n-k$ stabilizer generators. Let $\mathbf{s}\in F_q^{,n-k}$ denote the syndrome vector obtained by measuring all stabilizers, where each component $s_i$ is the outcome (in $\mathbb{F}q$) of the $i$th stabilizer measurement. An error can be represented by an operator $E$ acting on some subset of qudits; in the vector representation, we denote by $\mathbf{e}$ the error pattern (e.g., $e_j \in F_q$ indicates the magnitude of an $X$-type flip on qudit $j$, with $e_j=0$ meaning no error on that qudit). The syndrome is related to the error by a linear relation $\mathbf{s} = H,\mathbf{e}^T$, where $H$ is the parity-check matrix of the code (each row corresponds to a stabilizer generator) over $F_q$. In other words, $s_i = \sum{j} H{ij} e_j$ (mod $q$), capturing the fact that an error on qudit $j$ triggers the $i$th syndrome if qudit $j$ is involved in that stabilizer. We assume a stochastic error model where each physical qudit suffers an error independently. For example, a depolarizing noise on qudits can be modeled such that each qudit is hit by a non-identity Pauli with probability $p$ (with errors distributed uniformly over the $q^2-1$ Pauli errors on that qudit). For simplicity, one can consider separate $X$-type and $Z$-type errors; in a symmetric depolarizing model on qutrits ($q=3$), each qudit has probability $p$ of an $X^{a}Z^{b}$ error (with $a,b\in{0,1,2}$ not both zero) and $1-p$ of no error.

\subsection{RL-on-Greedy Decoding Procedure:} 
Algorithm \ref{alg:rl-greedy} presents the decoding procedure. The input is a raw syndrome $\mathbf{s}$ from the quantum code (assumed to arise from some physical error $\mathbf{e}$). First, a classical \textsc{GreedyDecode} is applied to $\mathbf{s}$. In this greedy decoder, at each step one chooses a corrective operation that locally reduces the syndrome weight. For instance, one rule is to flip a qudit involved in the largest number of unsatisfied stabilizers, updating the syndrome accordingly. The greedy stage halts when no single-qudit flip immediately improves the syndrome or when the syndrome becomes zero. Let $\mathbf{s}{\text{res}}$ be the residual syndrome after the greedy decoder. If $\mathbf{s}{\text{res}}=\mathbf{0}$, the decoder succeeds with the greedy correction alone. Otherwise, the residual syndrome $\mathbf{s}_{\text{res}}$ is handed to the RL agent for further decoding.

\begin{algorithm}[h]
\caption{RL-on-Greedy Decoder}
\label{alg:rl-greedy}
\begin{algorithmic}[1]
\Require Syndrome $\mathbf{s}$ for an $[[n,k,d]]q$ code, Greedy decoder $\mathcal{G}$, trained RL agent policy $\pi^*$
\Ensure Recovery operation $C$ that corrects $\mathbf{s}$
\State $C \gets \emptyset$ \Comment{Initialize empty correction}
\State $\mathbf{s}{\text{res}} \gets \mathcal{G}(\mathbf{s})$ \Comment{Apply Greedy decoder to get residual syndrome}
\State $C \gets C \cup$ (corrections applied by $\mathcal{G}$)
\If{$\mathbf{s}{\text{res}} = \mathbf{0}$}
\State \textbf{return} $C$ \Comment{Syndrome fully corrected by greedy stage}
\EndIf
\State $state \gets \mathbf{s}{\text{res}}$
\While{$state \neq \mathbf{0}$ \textbf{and} not exceeded max steps}
\State $a \gets \pi^*(state)$ \Comment{RL agent chooses an action (single-qudit flip)}
\State \textit{Apply} action $a$ to the code (flip the corresponding qudit)
\State $C \gets C \cup {a}$ \Comment{Accumulate the correction operation}
\State $state \gets$ new syndrome after applying $a$ \label{alg:synd-update}
\EndWhile
\State \textbf{return} $C$
\end{algorithmic}
\end{algorithm}

In words, the agent is only rewarded when it achieves the terminal goal of error correction (transition into the zero syndrome state). This sparse reward setup aligns with the idea that the agent “wins” by correcting all errors \cite{swe}. (In practice, one may also assign a small negative reward for each step to discourage overly long correction sequences, and a large negative reward if the agent exceeds the step limit or causes a logical failure.)

\subsection{Support for Nonbinary Codes:}
A key feature of our RL-on-Greedy decoder is that it generalizes to stabilizer codes over arbitrary finite fields $F_q$. The syndrome, state, and actions are all defined in the $q$-ary domain. In particular, syndrome values $s_i \in F_q$ may take values ${0,1,\dots,q-1}$ indicating the eigenvalue of each stabilizer generator. The greedy decoder can handle these multi-valued syndromes by choosing corrections that reduce the magnitude of syndrome components (modulo $q$). Likewise, the RL agent’s action space includes $q$-ary Pauli operations. For example, for a qudit code over $F_3$ (qutrits), the action set can include an $X$ flip (adding $+1$ mod 3 to a qudit’s $X$ error) and an $X^2$ flip (adding $-1$ mod 3) on each qudit (similarly for $Z$ flips). The state transition (Eq.,\ref{eq:syndrome-update}) inherently works mod $q$, so the RL agent learns directly on the $F_q$ syndrome space. This framework extends the benefits of machine-learned decoding to codes beyond qubits. Notably, our approach does not require any particular structure of the code – the RL agent is trained using only the syndrome update rules, which can be obtained from the stabilizer parity-check matrix of any $F_q$ code.

\subsection{Example: [[27,13,4]]$_3$ Qutrit Code – Performance Results:}
 To demonstrate the effectiveness of the RL-on-Greedy decoder, we evaluated it on a nonbinary stabilizer code with parameters $[[27,13,4]]_3$. This is a quantum code on $n=27$ qutrits (physical three-level systems) that encodes $k=13$ logical qutrits, with distance $d=4$. The distance $d=4$ implies the code can correct all single-qudit errors but not all pairs of errors (two errors can sometimes mimic a logical error) \cite{ket}. We consider a depolarizing noise model on qutrits: each qudit experiences an $X$, $Z$, or $X!Z$ error with total probability $p$ (each error type occurring with probability $p/3$ for symmetry), and no error with probability $1-p$. We compare two decoding strategies: a pure greedy decoder versus our RL-on-Greedy decoder (greedy + RL agent). Figure \ref{fig:rl_vs_greedy} plots the logical failure rate (block error probability) as a function of physical error rate $p$ for both decoders. As shown, the RL-enhanced decoder significantly outperforms the greedy approach across a range of $p$. For small $p$, both decoders correct most errors; but at moderate $p$ (on the order of a few percent), the greedy decoder’s failure rate rises sharply due to uncorrected two-qudit errors. The RL-on-Greedy decoder is able to resolve many of those cases, resulting in a markedly lower failure rate. For example, at $p=0.05$ (5$\%$ physical error rate), the greedy decoder fails on roughly a few $\times 10^{-2}$ fraction of trials, whereas the RL-on-Greedy decoder’s failure rate is about half of that. Even near the code’s break-even point, the RL agent recovers a significant portion of error patterns that confound the greedy approach. Importantly, RL-on-Greedy approaches the performance of an optimal decoder for this code: it corrects essentially all error patterns of weight $\le 2$ that are correctable in principle (given $d=4$), whereas the greedy decoder alone misses many of these. The result is an extended effective error-correction range and a higher effective noise threshold. We also observed that the RL agent typically requires only a small number of actions (often just 1 or 2 steps) to finish correcting the syndrome after the greedy stage, so the additional latency is minimal.

\begin{figure}[t]
\centering
\includegraphics[width=0.65\textwidth]{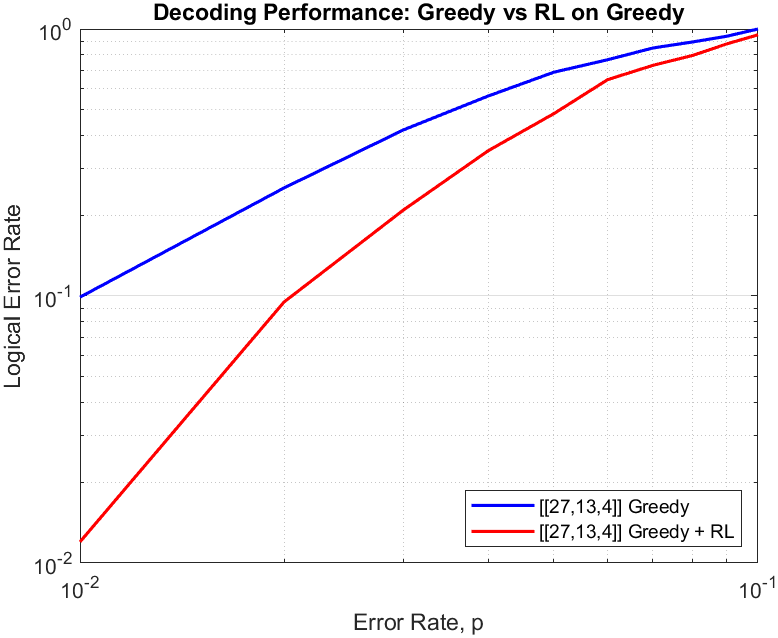}
\caption{Decoding failure rate vs. physical error probability for the $[[27,13,4]]_3$ qutrit code, comparing the classical Greedy decoder and the proposed RL-on-Greedy decoder. The RL-enhanced decoder achieves a lower failure rate at all depicted error probabilities. In particular, it corrects many two-error patterns that the Greedy decoder fails to correct, resulting in improved performance (lower logical error probability).}
\label{fig:rl_vs_greedy}
\end{figure}

\section{Discussion: }
These results illustrate that combining fast greedy decoding with a learned RL agent yields a powerful decoder that is both efficient and effective. The greedy first stage quickly fixes straightforward errors, dramatically reducing the syndrome weight, while the RL agent specializes in handling the “hard” residual configurations that require sequential or non-local tactics. The RL-on-Greedy decoder leverages the strengths of both approaches: it retains the low computational cost of greedy decoding and augments it with the adaptive, near-optimal decision-making of reinforcement learning. Notably, the RL agent is trained on the specific code and noise model, allowing it to discover decoding strategies tailored to that code’s structure (including any degeneracies or correlated syndrome patterns). Prior works have shown that RL decoders can achieve or surpass the performance of conventional decoders on various codes \cite{fit, swe}, and our findings reinforce this in the context of nonbinary codes. The general framework can be applied to other stabilizer codes, including larger distance codes or different finite-field alphabets, by training a suitable agent. This paves the way for flexible, high-accuracy decoders that can adapt to the detailed syndrome behavior of quantum codes. In summary, the RL-on-Greedy approach provides a scalable and broadly applicable enhancement to classical decoding algorithms, improving logical error rates without requiring substantial additional runtime. Its success on the $[[27,13,4]]_3$ code demonstrates the potential of combining traditional decoding heuristics with modern machine learning for quantum error correction.

\paragraph*{Acknowledgments.}
The first author was supported by TWAS/CNPq (Brazil) through fellowship No. 314966/2018-8. The implementation code for the reinforcement‐learning decoder introduced in this paper is publicly available at:\\
https://github.com/vnourozi/reinforcement-learning-decoder-for-AG-codes


\end{document}